\documentclass{llncs}
\usepackage{makeidx}
\usepackage{amsmath}
\usepackage[dvipsnames]{xcolor}
\usepackage{gastex}

\DeclareMathOperator{\Div}{Div}
\DeclareMathOperator{\Sep}{Sep}

\newenvironment{proofof}[1]{\noindent\textit{Proof
    \protect{#1}.}}

\begin{document}

\title{Regular Interval Exchange Transformations over a Quadratic Field}

\titlerunning{Regular IETs over a Quadratic Field}
\author{Francesco Dolce\inst{1}}
\institute{Universit\'e Paris Est}

\maketitle

\begin{abstract}
We describe a generalization of a result of Boshernitzan and Carroll: an extension of Lagrange's Theorem on continued fraction expansion of quadratic irrationals to interval exchange transformations. In order to do this, we use a two-sided version of the Rauzy induction.
In particular, we show that starting from an interval exchange transformation whose lengths are defined over a quadratic field and applying the two-sided Rauzy induction, one can obtain only a finite number of new transformations up to homothety.

\end{abstract}

\section{Introduction}
\label{sec:intro}
It is a truth universally acknowledged that the simple continued fraction expansion of a quadratic irrational must be eventually periodic (result known as Lagrange's Theorem).

Continued franctions are relatad to differents combinatorial tools, such as Stern-Brocot trees, mechanical words, rotations, etc. (see~\cite{concrete} and \cite{Lothaire2}).
An interesting representation of the continued fraction development is given by inducing the first return map of a $2$-interval exchange transformation, the ratio of whose lengths is a quadratic irrational, on the larger exchanged semi-interval.

Interval exchange transformations were introduced by Oseledec~\cite{Oseledec1966} following an earlier idea of Arnol'd~\cite{Arnold1963}. These transformations form a generalization of rotations of the circle (the two notions coincide when there are exactly $2$ intervals).
Rauzy introduced in~\cite{Rauzy1979} a transformation, now called Rauzy induction (or Rauzy-Veech induction), which operates on interval exchange transformations.
It actually transforms an interval exchange transformation into another, operating on a smaller semi-interval. Its iteration can be viewed, as mentioned, as a
generalization of the continued fraction development (since we work with $n$-interval exchange transformations with $n \geq 2$).
The induction consists in taking the first return map of the transformation with respect to a particular subsemi-interval of the original semi-interval.
A two-sided version of Rauzy induction is studied in~\cite{Twosided}, along with a characterization of the intervals reachable by the iteration of this two-sided induction, the so called admissible intervals.

Interval exchange transformations defined over quadratic fields have been studied by Boshernitzan and Carroll (\cite{Boshernitzan} and \cite{BoshernitzanCarroll1997}).
Under this hypothesis, they showed that, using iteratively the first return map on one of the semi-intervals exchanged by the transformation, one obtains only a finite number of different new transformations up to rescaling, extending the classical Lagrange's theorem that quadratic irrationals have a periodic continued fraction expansion.

In this paper we generalize this result, enlarging the family of transformations obtained using induction on every admissible semi-interval.
This contains the results of \cite{BoshernitzanCarroll1997} because every semi-interval exchanged by a transformation is admissible, while for $n>2$ there are admissible semi-intervals that we can not obtain using the induction only on the exchanged ones.

The paper is organised as follows.

In Section~\ref{sec:iet}, we recall some notions concerning interval exchange transformations, minimality and regularity. We also introduce an equivalence relation on the set of interval exchange transformations.
We finally recall the result of Keane~\cite{Keane1975} which proves that regularity is a sufficient condition for minimality of such a transformation (Theorem~\ref{theo:Keane}).

In Section~\ref{sec:rauzy} we recall the Rauzy induction and the generalization to its two-sided version. We also recall the definition of admissibility and how this notion is related to Rauzy induction (Theorems~\ref{theo:biRauzy1} and~\ref{theo:biRauzy2}). We conclude the section introducing the equivalence graph of a regular interval exchange transformation.

The final part of this paper, Section~\ref{sec:quadratic}, is devoted to the proof of our main result (Theorem~\ref{theo:quadratic}), i.e. the finiteness of the number of equivalence classes for a regular interval exchange transformation defined over a quadratic field.

\paragraph{Acknowledgements}
I would like to thank Dominique Perrin and Valerie Berth\'e for their suggestions and valuables remarks: \emph{merci !} Thanks also to Sonja M. Hiltunen for her corrections: \emph{kiitos!}
This work was supported by grants from R\'egion \^Ile-de-France.

\section{Interval Exchange Transformations}
\label{sec:iet}

Let us recall the definition of an interval exchange transformation (see~\cite{CornfeldFominSinai1982} or \cite{Twosided} for a more detailed presentation).

A \emph{semi-interval} is a nonempty subset of the real line of the form $[\ell,r[ = \{ z \in \bbbr \mid \ell \leq z < r \}$. Thus it
is a left-closed and right-open interval.

Let $A = \{ a_1, a_2, \ldots, a_s \}$ be a finite ordered alphabet with $a_1 < a_2 < \cdots < a_s$ and $(I_a)_{a \in A}$ an ordered partition of of $[\ell,r[$ in semi-intervals.
Set $\lambda_i$ the length of $I_{a_i}$.
Let $\pi \in \mathcal{S}_s$ be a permutation on $A$.

Define $\gamma_i=\sum_{a_j < a_i}\lambda_j$ and $\delta_{\pi(i)}=\sum_{\pi(a_j) < \pi(a_i)}\lambda_j$.
Set $\alpha_a=\delta_a - \gamma_a$.
The \emph{interval exchange transformation} relative to  $(I_a)_{a\in A}$ is the map $T:[\ell,r[\rightarrow [\ell,r[$ defined by

\begin{displaymath}
T(z)=z+\alpha_a\quad \text{ if } z\in I_a.
\end{displaymath}

Observe that the restriction of $T$ to $I_a$ is a translation onto $J_a=T(I_a)$, that $\gamma_i$ is the left boundary of $I_{a_i}$ and that $\delta_j$ is the left boundary of $J_{a_i}$.

Note that the family $(J_a)_{a\in A}$ is also a partition of $[\ell,r[$.
In particular, the transformation $T$ defines a bijection from $[\ell,r[$ onto itself.

An interval exchange transformation relative to $(I_a)_{a\in A}$ is also called a $s$-interval exchange transformation.
The values $(\alpha_a)_{a\in A}$ are called the \emph{translation values} of the transformation $T$.
We will also denote $T = T_{\pi, \lambda}$, where $\lambda = (\lambda_i)_{a_i \in A}$ is the ordered sequence of lengths of the semi-intervals.

\begin{example}
\label{ex:rotation}
Let $T = T_{\pi, \lambda}$ be the interval exchange transformation corresponding to $A=\{a,b\}$,
$a < b$, $\pi = (12)$, i.e. such that $\pi(b) < \pi(a)$
and $\lambda = (1-\alpha, \alpha)$ with $\alpha = \frac{3-\sqrt{5}}{2}$.
Thus the two semi-intervals exchanged by $T$ are
$I_a=[0,1-\alpha[$ and $I_b=[1-\alpha,1[$.
The transformation $T$, representend in Figure~\ref{fig:rotation}, is the rotation of angle $\alpha$ on the semi-interval $[0,1[$ defined by
$T(z)=z+\alpha\bmod 1$.

\begin{figure}[hbt]
\centering
\gasset{Nadjust=wh,AHnb=0}
\begin{picture}(100,10)(0,0)

\node[fillcolor=red](0H)(0,10){}
\node[Nframe=n](0)(0,13){$0$}
\node[fillcolor=blue](1-alphaH)(61.8,10){}
\node[Nframe=n](1-alpha)(61.8,13){$1-\alpha$}
\node(1H)(100,10){}
\node[Nframe=n](1)(100,13){$1$}

\node[fillcolor=blue](0B)(0,0){}
\node[fillcolor=red](alphaB)(38.2,0){}
\node[Nframe=n](alpha)(38.2,3){$\alpha$}
\node(1B)(100,0){}

\drawedge[linecolor=red](0H,1-alphaH){}
\drawedge[linecolor=blue](1-alphaH,1H){}
\drawedge[linecolor=blue](0B,alphaB){}
\drawedge[linecolor=red](alphaB,1B){}

\end{picture}
\caption{Rotation of angle $\alpha$ on the semi-interval $[0,1[$.}
\label{fig:rotation}
\end{figure}
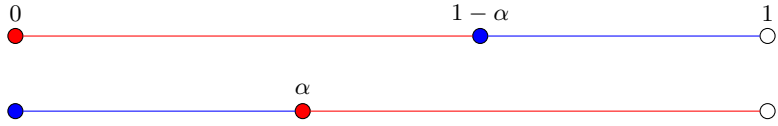
\end{example}

Note that the transformation $T_{\pi, \lambda,}$ does not depend on the relative position, (the choice of the left point $\ell$).

It is easy to verify that the family of $s$-interval exchange transformations is closed by taking inverses.

\subsection{Equivalent Interval Exchange Transformations}
\label{subsec:equivalent}

Two $s$-interval exchange transformation $T = T_{\pi, \lambda}$ and $S = T_{\sigma, \mu}$ are said to be \emph{equivalent} either if $\sigma = \pi$ and $\mu = c\lambda$ for some $c > 0$ or if $\sigma = \tau \circ \pi$ and $\mu = c \widetilde{\lambda}$, where $\tau : i \mapsto (s-i+1)$ is the permutation that reverses the names of the semi-intervals and $\widetilde{\lambda} = (\lambda_s, \lambda_{s-1}, \ldots, \lambda_1)$.

We denote by $[T_{\pi, \lambda}]$ the equivalence class of $T_{\pi, \lambda}$.

\begin{example}
\label{ex:equivalentiet}
Let $T_{\pi, \mu}$ be the interval exchange transformation defined by $\pi = (12)$ and $\mu = (1-2\alpha, \alpha)$, with $\alpha = \frac{3-\sqrt{5}}{2}$ (see Figure~\ref{fig:psirotation}). The transformation $T_{\pi, \mu}$ is equivalent to the transformation $T_{\pi,\lambda}$ of Example~\ref{ex:rotation}.
Indeed $\alpha^2 = 3\alpha-1$ and one can easily show that $\mu = (1-\alpha)\widetilde{\lambda}$.

\begin{figure}[hbt]
\centering
\gasset{Nadjust=wh,AHnb=0}
\begin{picture}(61.8,10)(0,0)

\node[fillcolor=red](0H)(0,10){}
\node[Nframe=n](0)(0,13){$0$}
\node[fillcolor=blue](1-2alphaH)(23.6,10){}
\node[Nframe=n](1-2alpha)(23.6,13){$1-2\alpha$}
\node(1-alphaH)(61.8,10){}
\node[Nframe=n](1-alpha)(61.8,13){$1-\alpha$}

\node[fillcolor=blue](0B)(0,0){}
\node[fillcolor=red](alphaB)(38.2,0){}
\node[Nframe=n](1-alpha)(38.2,3){$\alpha$}
\node(1-alphaB)(61.8,0){}

\drawedge[linecolor=red](0H,1-2alphaH){}
\drawedge[linecolor=blue](1-2alphaH,1-alphaH){}
\drawedge[linecolor=blue](0B,alphaB){}
\drawedge[linecolor=red](alphaB,1-alphaB){}

\end{picture}
\caption{Transformation $T_{(12),(1-2\alpha, \alpha)}$.}
\label{fig:psirotation}
\end{figure}
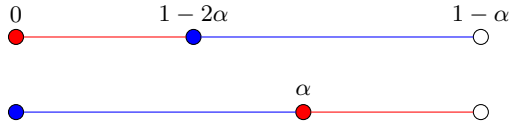
\end{example}

\subsection{Regular Interval Exchange Transformations}
\label{subsec:regular}

The \emph{orbit} of a point $z \in [\ell,r[$ is the set $\mathcal{O}(z) = \{T^n(z)\mid n\in\bbbz \}$. 
The transformation $T$ is said to be \emph{minimal} if for any $z\in[\ell, r[$, $\mathcal{O}(z)$ is dense in $[\ell,r[$.

The points $0=\gamma_1,\gamma_2,\ldots,\gamma_s$ form the set of \emph{separation points} of $T$, denoted $\Sep(T)$.
Note that the transformation $T$ has at most $s-1$ \emph{singularities} (points at which it is not continuous), which are
among the nonzero separation points $\gamma_2,\ldots,\gamma_s$.

An interval exchange transformation $T_{\pi,\lambda}$ is called \emph{regular} if the orbits of the nonzero separation points $\gamma_2,\ldots,\gamma_s$ are infinite and disjoint.
Note that the orbit of $0$ cannot be disjoint from the others since one has $T(\gamma_i)=0$ for some $i$ with $2\le i\le s$.

A regular interval exchange transformation is also said to satisfy the \emph{idoc} (infinite disjoint orbit condition). It is also said to have the Keane property or to be without \emph{connection} (see~\cite{BoissyLanneau2009}).

Note that  since $\delta_{\pi(2)}=T(\gamma_2),\ldots,\delta_{\pi(s)}=T(\gamma_s)$, $T$ is regular if and only if the orbits of $\delta_{\pi(2)},\ldots,\delta_{\pi(s)}$
are infinite and disjoint.

As an example, the $2$-interval exchange transformation of Example~\ref{ex:rotation} is regular, as every rotation of irrational angle.
The following result is due to Keane~\cite{Keane1975}.

\begin{theorem}[Keane]
\label{theo:Keane}
A regular interval exchange transformation is minimal.
\end{theorem}

The converse is not true. Indeed, consider the rotation of angle $\alpha$ with $\alpha$ irrational, as a $3$-interval exchange transformation
with $\lambda=(1-2\alpha,\alpha,\alpha)$ and $\pi=(123)$.
The transformation is minimal, as is any rotation of an irrational angle, but it is not regular since $\gamma_2=1-2\alpha$, $\gamma_3=1-\alpha$ and
thus $\gamma_3=T(\gamma_2)$.

\section{Rauzy Induction}
\label{sec:rauzy}

We recall in this section the transformation called Rauzy induction, defined in~\cite{Rauzy1979}, which operates on regular interval
transformations, and some results concerning this transformation (Theorems~\ref{theo:Rauzy1} and \ref{theo:Rauzy2}).
We also recall a two-sided version of this transformation studied in~\cite{Twosided} and some of the results relative to it (Theorems~\ref{theo:biRauzy1} and~\ref{theo:biRauzy2}).

\subsection{Right Rauzy Induction}
\label{subsec:right}

Let $T=T_{\pi, \lambda}$ be an interval exchange transformation relative to  $(I_a)_{a\in A}$. 

For  $\ell<t<r$, the semi-interval $[\ell,t[$  is \emph{right admissible}  for $T$ if there is a $k\in\bbbz$ such that $t=T^k(\gamma_a)$ for some $a\in A$
and
\begin{enumerate}
\item[(i)] if $k>0$, then $t<T^h(\gamma_a)$ for all $h$ such that $0<h<k$,
\item[(ii)] if $k\le 0$, then $t<T^h(\gamma_a)$ for all $h$ such that $k<h\le 0$.
\end{enumerate}
We also say that $t$ itself is right admissible.
Note that all semi-intervals $[\ell,\gamma_a[$ with $\ell<\gamma_a$ are right admissible.
Similarly all semi-intervals $[\ell, \delta_a[$ with $\ell<\delta_a$ are right admissible.

\begin{example}
\label{ex:division}
Let $T$ be the interval exchange transformation of Example~\ref{ex:rotation}.
The semi-interval $[0,t[$ for $t=1-\alpha$ or $t=1-2\alpha$ is right admissible since $1-\alpha=\gamma_2$ and $1-2\alpha=T^{-1}(\gamma_2) < \gamma_2$.
On the contrary, for $t=2-3\alpha$, it is not right admissible because $t=T^{-2}(\gamma_2)$ but $\gamma_2<t$ contradicting (ii).
\end{example}

Assume now that $T$ is minimal.
Let $I\subset [\ell,r[$ be a  semi-interval.
Since $T$ is minimal, for each $z\in [\ell,r[$ there exists an integer $n>0$ such that $T^n(z)\in I$.

The \emph{transformation induced} by $T$ on $I$ is the transformation $S:I\rightarrow I$ defined for $z\in I$ by
$S(z)=T^n(z)$ with $n=\min\{n>0 \mid T^n(z)\in I\}$.
The semi-interval $I$ is called the \emph{domain} of $S$, denoted $D(S)$.

\begin{example}
Let $T$ be the transformation of Example~\ref{ex:rotation}.
Let $I=[0,1-\alpha[$. 
The transformation induced by $T$ on $I$ is
\begin{displaymath}
S(z)=
\begin{cases}
T(z)	&	\text{ if $0\le z< 1-2\alpha$} \\
T^2(z)	&	\text{ otherwise}.
\end{cases}
\end{displaymath}
\end{example}

The following result is Theorem 14 in~\cite{Rauzy1979}.

\begin{theorem}[Rauzy]
\label{theo:Rauzy1}
Let $T$ be a regular $s$-interval exchange transformation and let $I$ be a right admissible interval for $T$.
The transformation induced by $T$ on $I$ is a regular $s$-interval exchange transformation.
\end{theorem}

Note that the transformation induced by an $s$-interval exchange transformation on $[\ell,r[$ on any
semi-interval included in $[\ell,r[$ is always an interval exchange transformation on at most $s+2$ intervals (see~\cite{CornfeldFominSinai1982}, Chapter 5 p. 128).

\begin{example}
\label{ex:psirotation}
Consider again the transformation of Example~\ref{ex:rotation}.
The transformation induced by $T$ on the semi-interval $I=[0,1-\alpha[$ is the $2$-interval exchange transformation represented in Figure~\ref{fig:psirotation}.
\end{example}

Let $T=T_{\pi,\lambda}$ be a regular s-interval exchange transformation on $[\ell,r[$.
Set
\begin{displaymath}
Z(T)=[\ell,\max\{\gamma_{s},\delta_{\pi(s)}\}[.
\end{displaymath}
Note that $Z(T)$ is the largest semi-interval which is right-admissible  for $T$.
We denote by $\psi(T)$ the transformation induced by $T$ on $Z(T)$.

The following result is Theorem 23 in~\cite{Rauzy1979}.

\begin{theorem}[Rauzy]
\label{theo:Rauzy2}
Let $T$ be a regular interval exchange transformation.
A semi-interval $I$ is right admissible for $T$ if and only if there exists an integer $n\ge 0$ such that $I=Z(\psi^n(T))$. In this case, the transformation induced by $T$ on $I$ is $\psi^{n+1}(T)$.
\end{theorem}

The map $T\mapsto\psi(T)$ is called the \emph{right Rauzy induction}.

\begin{example}
\label{ex:induced}
Consider again the transformation $T$ of Example~\ref{ex:rotation}.
Since $Z(T)=[0,1-\alpha[$, the transformation $\psi(T)$ is the one represented in Figure~\ref{fig:psirotation}.
\end{example}

\subsection{Left Rauzy Induction}
\label{subsec:left}

The symmetrical notion of \emph{left Rauzy induction} is defined similarly.
Define
\begin{displaymath}
Y(T)=[\min\{\gamma_2,\delta_{\pi(2)}\},r[.
\end{displaymath}
We denote by $\varphi(T)$ the transformation induced by $T$ on $Y(T)$. The map $T~\mapsto~\varphi~(T)$ is called the \emph{left Rauzy induction}. 

The notion of left admissible interval is symmetrical to that of right admissible.
For  $\ell<t<r$, the semi-interval $[t,r[$  is \emph{left admissible}  for $T$ if there is a $k\in\bbbz$ such that $t=T^k(\gamma_a)$ for some $a\in A$
and
\begin{enumerate}
\item[(i)] if $k>0$, then $T^h(\gamma_a)<t$ for all $h$ such that $0<h<k$,
\item[(ii)] if $k\le 0$, then $T^h(\gamma_a)<t$ for all $h$ such that $k<h\le 0$.
\end{enumerate}
We also say that $t$ itself is left admissible.

The symmetrical statement of Theorem~\ref{theo:Rauzy2} also holds
for left admissible intervals.
Note that, similar to the right admissibility, we have $[\gamma_a, r[$ and $[\delta_a, r[$ left admissible for every $a \in A$.

\begin{example}
\label{ex:phirotation}
Let $T$ be the transformation of Example~\ref{ex:rotation}.
One has $Y(T) = [\alpha, 1[$. The transformation $\varphi(T) = T_{(12),(1-2\alpha, 1-\alpha)}$ is represented in Figure~\ref{fig:phirotation}.

\begin{figure}[hbt]
\centering
\gasset{Nadjust=wh,AHnb=0}
\begin{picture}(61.8,10)(38.2,0)

\node[fillcolor=red](alphaH)(38.2,10){}
\node[Nframe=n](alpha)(38.2,13){$\alpha$}
\node[fillcolor=blue](1-alphaH)(61.8,10){}
\node[Nframe=n](1-alpha)(61.8,13){$1-\alpha$}
\node(1H)(100,10){}
\node[Nframe=n](1)(100,13){$1$}

\node[fillcolor=blue](alphaB)(38.2,0){}
\node[fillcolor=red](2alphaB)(76.4,0){}
\node[Nframe=n](2alpha)(76.4,3){$2\alpha$}
\node(1)(100,0){}

\drawedge[linecolor=red](alphaH,1-alphaH){}
\drawedge[linecolor=blue](1-alphaH,1H){}
\drawedge[linecolor=blue](alphaB,2alphaB){}
\drawedge[linecolor=red](2alphaB,1B){}

\end{picture}
\caption{Transformation $T_{(12),(1-2\alpha, 1-\alpha)}$ induced by $T$ on $[\alpha, 1[$.}
\label{fig:phirotation}
\end{figure}
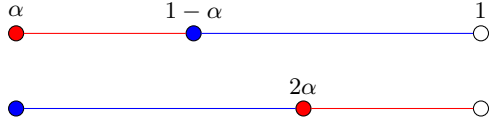
\end{example}

Note that for a $2$-interval exchange transformation $T$, one has $\psi(T) = \varphi(T)$, whereas in general the two transformations are different.

\subsection{Two-sided Induction}
\label{subsec:two}

In this section, we generalize the left and right Rauzy inductions to a two-sided induction (see~\cite{Twosided} for a detailed presentation).

Let $T=T_{\pi, \lambda}$ be an $s$-interval exchange transformation on $[\ell,r[$ relative to $(I_a)_{a\in A}$.
For a semi-interval $I=[u,v[\subset[\ell,r[$, we define the following functions
on $[\ell,r[$
\begin{displaymath}
\rho^+_{I,T}(z)=\min\{n>0\mid T^n(z)\in\ ]u,v[\},\quad 
\rho^-_{I,T}(z)=\min\{n\ge 0\mid T^{-n}(z)\in\ ]u,v[\}.
\end{displaymath}

We define the set of \emph{neighbours} of $z$ with respect to $I$ and $T$ as
\begin{displaymath}
N_{I,T}(z)=\{ T^k(z) \mid -\rho_{I,T}^-(z)\le k < \rho_{I,T}^+(z) \}.
\end{displaymath}

The set of \emph{division points} of $I$ with respect to $T$ is the finite set
\begin{displaymath}
\Div(I,T)=\bigcup_{i=1}^s N_{I,T}(\gamma_i).
\end{displaymath}

For $\ell\le u<v\le r$, we say that the semi-interval $I=[u,v[$ is \emph{admissible} for $T$ if $u,v\in\Div(I,T)\cup \{r\}$.

Note that a semi-interval $[\ell,v[$ is right admissible if and only if  it is admissible and that a semi-interval $[u,r[$ is left admissible if and only if it is admissible.
Note also that $[\ell,r[$ is admissible.

Note also that for a regular interval exchange transformation relative to a partition $(I_a)_{a\in A}$, each of the semi-intervals
$I_a$ (or $J_a$) is admissible although only the first one is right admissible (and the last one is left admissible).

Recall that $\Sep(T)$ denotes the set of separation points of $T$, i.e. the points $\gamma_1=0,\gamma_2,\ldots,\gamma_s$ (which are the left boundaries of the semi-intervals $I_{a_1}, I_{a_2}, \ldots, I_{a_s}$).
The following generalization of Theorem~\ref{theo:Rauzy1} is proved in~\cite{Twosided}.

\begin{theorem}
\label{theo:biRauzy1}
Let $T$ be a regular $s$-interval exchange transformation on $[\ell,r[$.
For any admissible semi-interval $I=[u,v[$, the transformation $S$ induced by $T$ on $I$ is a regular $s$-interval exchange transformation with separation points $\Sep(S)=\Div(I,T)\cap I$.
\end{theorem}

We have already noted that for any $s$-interval exchange transformation  on $[\ell,r[$ and any semi-interval $I$ of $[\ell,r[$, the transformation $S$
induced by $T$ on $I$ is an interval exchange transformation on at most $s+2$-intervals. Actually, it follows from the proof of Lemma 2, page 128
in~\cite{CornfeldFominSinai1982} that, if $T$ is  regular and $S$ is an $s$-interval exchange transformation with separation points $\Sep(S)=\Div(I,T)\cap I$, then $I$ is admissible.
Thus the converse of Theorem~\ref{theo:biRauzy1} is also true.

The following generalization of Theorem~\ref{theo:Rauzy2} is proved in~\cite{Twosided}.

\begin{theorem}
\label{theo:biRauzy2}
Let $T$ be a regular $s$-interval exchange transformation on $[\ell,r[$.
A semi-interval $I$ is admissible for $T$ if and only if there exists a sequence $\chi\in\{\varphi,\psi\}^*$ such that $I$ is the domain of $\chi(T)$. In this case, the transformation induced by $T$ on $I$ is $\chi(T)$.
\end{theorem}

\subsection{Equivalence Graph}
\label{subsec:graph}

For an interval exchange transformation $T$ we consider the directed graph $G(T)$, called the \emph{equivalence graph} of $T$, defined as follows. The vertices are the equivalence classes of transformations obtained starting from $T$ and applying all possible $\chi \in \left\{ \psi, \varphi \right\}^*$.
There is an edge starting from a vertex $[T]$ to a vertex $[S]$ if and only if $S = \theta(T)$ for two transformations $T \in [T]$ and $S \in [S]$ and a $\theta \in \{ \psi, \varphi \}$.

\begin{example}
\label{ex:equivalence}
Let $T = T_{(12), \left(1-\alpha, \alpha\right)}$ be the regular $2$-interval exchange transformation of Example~\ref{ex:rotation}.
Applying the right and the left Rauzy induction on $T$ we obtain
$\psi(T) = \varphi(T) = T_{(12), \left(1-2\alpha,\alpha\right)}$ (see Examples~\ref{ex:psirotation}). These two transforations are equivalent (see Example~\ref{ex:equivalentiet}). Therefore the equivalence graph of $T$, represented in Figure~\ref{fig:equivalence}, contains only one vertex.
Note that the ratio of the two lengths of the semi-intervals exchanged by $T$ is $\frac{1 - \alpha}{\alpha} = \frac{1 + \sqrt{5}}{2} = \phi = 1 + \frac{1}{1 + \frac{1}{1+ \cdots}} = 1 + \frac{1}{\phi}$, i.e. the golden ratio.

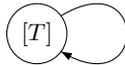
\begin{figure}[hbt]
\centering
\gasset{Nframe=y}
\begin{picture}(10,5)(0,-5)

\node(T)(0,0){$[T]$}

\drawloop[curvedepth=5,loopangle=0](T){}

\end{picture}
\caption{Equivalence graph of the transformation $T = T_{(12),\left(1-\alpha,\alpha\right)}$.}
\label{fig:equivalence}
\end{figure}
\end{example}
Note that, in general, the equivalence graph can be infinite. We will give in the next section a sufficient condition for the equivalence graph to be finite.

\section{Interval Exchange Transformations Over a Quadratic Field}
\label{sec:quadratic}
An interval exchange transformation is said to be defined over a set $Q \subset \bbbr$ if the lengths of all exchanged intervals belong to $Q$.

The following is proved in~\cite{BoshernitzanCarroll1997}.
Let $T$ be a minimal interval exchange transformation on semi-intervals defined over a quadratic number field.
Let $(T_n)_{n\ge 0}$ be a sequence of interval exchange transformation such that $T_0=T$ and $T_{n+1}$ is the transformation induced by $T_n$ on one of its exchanged semi-intervals $I_n$. Then, up to rescaling all intervals $I_n$ to the same length, the sequence $(T_n)$ contains finitely many distinct transformations.

In this section we generalize this results and prove that, under the above hypothesis on the lengths of the semi-intervals and up to rescaling, there are finitely many transformations obtained by the two-sided Rauzy induction.

\begin{theorem}
\label{theo:quadratic}
Let $T$ be a regular interval exchange transformation defined over a quadratic field.
The family of all induced transformation of $T$ over an admissible semi-interval contains finitely many distinct transformations up to equivalence.
\end{theorem}

Note that the previous theorem implies the result of \cite{BoshernitzanCarroll1997}. Indeed every semi-interval exchanged by a transformation is admissible, while for $n>2$ there are admissible semi-intervals that we can not obtain using the induction only on the exchanged ones.

The proof of the Theorem~\ref{theo:quadratic} is based on the fact that for each minimal interval exchange transformation defined over a quadratic field, a certain measure of the arithmetic complexity of the admissible semi-intervals is bounded.

\subsection{Complexities}
\label{subsec:complexities}

Let $T$ be an interval exchange transformation on a semi-interval $[\ell, r[$ defined over a quadratic field $ \bbbq[\sqrt{d}]$, where $d$ is a square free integer $\geq 2$.
Without loss of generality, one may assume, by replacing $T$ by an equivalent interval exchange transformation if necessary, that $T$ is defined over the ring
$ \bbbz[\sqrt{d}] = \{ m + n\sqrt{d} \,\, | \,\, m,n \in \bbbz \} $
and that all $\gamma_i$ and $\alpha_i$ lie in $\bbbz[\sqrt{d}]$ (replacing $[\ell, r[$ if necessary by its equivalent translate with $\gamma_0 = \ell \in \bbbz[\sqrt{d}]$).

For $z = m + n \sqrt{d}$ let define
$\Psi(z) = \max ( |m|, |n| )$.

The following Proposition is proved in~\cite[Proposition 2.2]{BoshernitzanCarroll1997}.

\begin{proposition}
\label{pro:o1}
For every $z \in \bbbz[\sqrt{d}] \setminus \{ 0 \}$, one has $|z| \, \Psi(z) > \frac{1}{2 \sqrt{d}}$.
\end{proposition}

Let $\mathcal{A}([\ell,r[)$ be the algebra of subsets $S \subset [\ell,r[$ which are finite unions $S = \bigcup_j I_j$ of semi-intervals defined over $\bbbz[\sqrt{d}]$, i.e. $I_j = [\ell_j, r_j[$ for some $\ell_j, r_j \in \bbbz[\sqrt{d}]$.
Note that the algebra $\mathcal{A}([\ell,r[)$ is closed under taking finite unions, intersections and passing to complements in $[\ell, r[$.

We define the \emph{complexity} $\Psi(S)$ and the \emph{reduced complexity} $\Pi(S)$ of a subset $S \in \mathcal{A}([\ell,r[)$ as
$$ \Psi(S) = \max \{ \Psi(z) \,\, | \,\, z \in \partial(S) \} \quad \mbox{ and } \quad \Pi(S) = |S| \, \Psi(S), $$
where $\partial(S)$ is the boundary of $S$ and $|S|$ stands for the Lebesgue measure of $S$.

A key tool to prove Theorem~\ref{theo:quadratic} is the following Theorem proved in~\cite[Theorem 3.1]{BoshernitzanCarroll1997}.

\begin{theorem}[Boshernitzan]
\label{theo:bosh}
Let $T$ be a minimal interval exchange transformation on an interval $[\ell,r[$ defined over a quadratic number field. Assume that $(Y_n)_{\geq 1}$ is a sequence of semi-intervals of $[\ell,r[$ such that the set $\{ \Pi(Y_n) \,\, | \,\, n \geq 1 \}$ is bounded.
Then the sequence $S_n$ of interval exchange transformations obtained by inducing $T$ on $Y_n$ contains finitely many distinct equivalence classes of interval exchange transformations.
\end{theorem}

The following Proposition is proved in~\cite[Proposition 2.1]{BoshernitzanCarroll1997}. It shows that the complexity of a subset $S$ and of its image $T(S)$ differs at most by a constant that depends only on $T$.

\begin{proposition}
\label{pro:psiT-psi}
There exists a constant $u = u(T)$ such that for every $S \in \mathcal{A}([\ell,r[)$ and $z \in [\ell,r[$ one has
$$| \Psi(T(S)) - \Psi(S) | \leq u \quad \mbox{ and } \quad \Psi(T(z) - z) \leq u.$$
Moreover, one has $\Psi(\gamma) \leq u$ and $\Psi(T(\gamma)) \leq u$ for every separation point $\gamma$.
\end{proposition}

Clearly, by Proposition~\ref{pro:psiT-psi}, one also has $| \Psi(T^{-1}(S)) - \Psi(S) | \leq u$ for every $S \in \mathcal{A}([\ell,r[)$ and $\Psi(T^{-1}(z) - z) \leq u$ for every $z \in [\ell,r[$.

Combining Proposition~\ref{pro:o1} and Proposition~\ref{pro:psiT-psi} we easily obtain the following (see also~\cite[Corollary 2.3]{BoshernitzanCarroll1997}).
\begin{corollary}
\label{cor:cn}
There exists a constant $c > 0$ such that for every $n \in \bbbz$ and $z \in [\ell, r[$ one has either $T^n(z) = z$ or $|T^n(z) - z| > \frac{c}{n}$.
\end{corollary}

The following Proposition, proved in~\cite[Proposition 2.4]{BoshernitzanCarroll1997}, determines a lower bound on the reduced complexity of a nonempty subset $S \in \mathcal{A}([\ell, r|)$.

\begin{proposition}
\label{pro:piS}
Let $S \in \mathcal{A}([\ell,r[)$ be a subset composed of $n$ disjoint semi-intervals. Then
$ \Pi(S) > \frac{n}{4 \sqrt{d}}$.
\end{proposition}

\subsection{Return Times}
\label{subsec:return}

Let $T$ be an interval exchange transformation.
For a subset $S \in \mathcal{A}([\ell, r[)$ we define the maximal positive and maximal negative return times for $T$ on $S$ by the functions
$$
\rho^+(S) = \min \left\{ n \geq 1 \, | \, T^n(S) \subset \bigcup_{i = 0}^{n-1} T^i(S) \right\},
$$
and
$$
\rho^-(S) = \min \left\{ m \geq 1 \, | \, T^m(S) \subset \bigcup_{i = 0}^{m-1} T^{-i}(S) \right\}.
$$

We also define the minimal positive and the minimal negative return times as
$$
\sigma^+(S) = \min \left\{ n \geq 1 \, | \, T^n(S) \cap S \neq \emptyset \right\},
$$
and
$$
\sigma^-(S) = \min \left\{ m \geq 1 \, | \, T^{-m}(S) \cap S \neq \emptyset \right\},
$$

Note that, when $S$ is a semi-interval, we have $\rho^+(S) = \max\limits_{z \in S} \rho_{S,T}^+(z)$ and $\sigma^+(S) = \min\limits_{z \in S} \rho_{S,T}^+(z)$.
Symmetrically $\rho^-(S) = \max\limits_{z \in S} \rho_{S,T}^-(z) + 1$ and $\sigma^-(S) = \min\limits_{z \in S} \rho_{S,T}^-(z) + 1$.

If $T$ is minimal, it is clear that
$ [\ell, r[ \, = \bigcup_{i=0}^{\rho^+(S)-1} T^i(S) = \bigcup_{i=0}^{\rho^-(S)-1} T^{-i}(S)$.

Let $\zeta, \eta$ be two functions. We write $\zeta = O(\eta)$ if there exists a constant $C$ such that $|\zeta| \leq C |\eta|$. We write $\zeta = \Theta(\eta)$ if one has both $\zeta = O(\eta)$ and $\eta = O(\zeta)$.
Note that $\Theta$ is an equivalence relation.

Boshernitzan and Carroll give in~\cite{BoshernitzanCarroll1997} two upper bounds for $\rho^+(S)$ and $\sigma^+(S)$ for a subset $S$ (Theorems 2.5 and 2.6 respectively) and a more precise estimation when the subset is a semi-interval (Theorem 2.8). Some slight modifications of the proofs can be made so that the results hold also for $\rho^-$ and $\sigma^-$. We summarize these estimates in the following theorem.

\begin{theorem}
\label{theo:sim}
For every $S \in \mathcal{A}([\ell,r[)$ one has
$ \rho^+(S), \rho^-(S) = O(\Psi(S))$
and
$\sigma^+(S), \sigma^-(S) = O\left(\frac{1}{| S |}\right)$.
Moreover, if $T$ is minimal and $J$ is a semi-interval, then
$\rho^+(J) = \Theta\left(\rho^-(J)\right) = \Theta\left(\sigma^+(J)\right) = \Theta\left(\sigma^-(J)\right) = \Theta\left(\frac{1}{| J |}\right)$.
\end{theorem}

An immediate corollary of Theorem~\ref{theo:sim} is the following (see also Corollary 2.9 of~\cite{BoshernitzanCarroll1997}).
\begin{corollary}
\label{cor:oM}
Let $T$ be minimal and assume that
$$
\{ T^i(z) \, | \, -m+1 \leq i \leq n-1 \} \cap J = \emptyset
$$
for some point $z \in [\ell,r[$, some semi-interval $J \subset [\ell,r[$ and some integers $m, n \geq 1$.
Then $|J| = O\left( \frac{1}{\max \{ m,n \}} \right)$
\end{corollary}
\begin{proof}
By the hypothesis, $z \notin \bigcup_{i=0}^{n-1} T^{-i}(J)$ we have $\rho^-(J) \geq n$.
Then, using Theorem~\ref{theo:sim}, we obtain
$ |J| = \Theta\left(\frac{1}{\rho^-(J)}\right) = O\left( \frac{1}{n} \right)$.
Symmetrically, since $\rho^+(J) \geq m$, one has $ |J| = O\left( \frac{1}{m} \right)$.
Then
$ |J| = O\left( \min \left\{ \frac{1}{m}, \frac{1}{n} \right\} \right) = O\left( \frac{1}{\max \{ m,n \}} \right)$.
\qed
\end{proof}

\subsection{Reduced Complexity of Admissible Semi-Intervals}
\label{subsec:main}

In order to demonstrate the main theorem (Theorem~\ref{theo:quadratic}), we prove some preliminary results concerning the reduced complexity of admissible semi-intervals.

Let $T$ be an $s$-interval exchange transformation. Recall that we denote by $\Sep(T) = \{ \gamma_i \, | \, 0 \leq i \leq s-1 \}$ the set of separation points.
For every $n \geq 0$ define
$\mathcal{D}_n(T) = \bigcup_{i=0}^{n-1} T^{-i} \big(\Sep(T)\big)$
with the convention $\mathcal{D}_0 = \emptyset$.

Since $\Sep(T^{-1}) = T\big(\Sep(T)\big)$, one has $\mathcal{D}_n(T^{-1}) = T^{n-1}\big(\mathcal{D}_n(T)\big)$.

Given two integers $m,n \geq 1$, we can define $\mathcal{D}_{m,n} = \mathcal{D}_m(T) \cup \mathcal{D}_n(T^{-1})$.
An easy calculation shows that
$$
\mathcal{D}_{m,n}(T) = \bigcup_{i=-m+1}^{n} T^i\big(\Sep(T)\big).
$$

Observe also that $\mathcal{D}_{m,n}(T) = T^{n}\big( \mathcal{D}_{m+n}(T)\big) = T^{-m+1}\big( \mathcal{D}_{m+n}(T)\big)$.

Denote by $\mathcal{V}_{m,n}(T)$ the family of semi-intervals whose endpoints are in $\mathcal{D}_{m,n}(T)$.
Put $\mathcal{V}(T) = \bigcup_{m,n \geq 0} \mathcal{V}_{m,n}(T)$.

Every admissible semi-interval belong to $\mathcal{V}(T)$, while the converse is not true.

\begin{theorem}
\label{theo:pi}
$\Pi(J) = \Theta(1)$ for every semi-interval $J$ admissible for $T$.
\end{theorem}
\begin{proof}
Let $m, n$ be the two minimal integers such that $J = [t,w[ \, \in \mathcal{V}_{m,n}(T)$.
Then $t,w \in \{T^m(\gamma_i) \, | \, 1 \leq i \leq s \} \cup \{T^{-n}(\gamma_i) \, | \, 1 \leq i \leq s \}$.
Suppose, for instance, $t = T^M(\gamma)$, with $M = \max \{m,n\}$ and $\gamma$ a separation point.
The other cases -- i.e. $t = T^{-M}(\gamma)$, $w=T^{M}(\gamma)$ or $w=T^{-M}(\gamma)$ -- are proved similary.

The only semi-interval in $\mathcal{V}_{0,0}(T)$ is $[\ell,r[$ and clearly in this case the theorem is verified.
Suppose then that $J \in \mathcal{V}_{m,n}(T)$ for some non-negative integers $m,n$ with $m+n > 0$.

We have
$\Psi(J) = \max \{ \Psi(t), \Psi(w) \} \leq Mu$ where $u$ is the constant introduced in Proposition~\ref{pro:psiT-psi}.

Moreover, by the definition of admissibility one has
$ \{ T^j(\gamma) \, | \, 1 \leq j \leq M \} \cap J = \emptyset$.
Thus, by Corollary~\ref{cor:oM} we have $|J| = O(\frac{1}{M})$.
Then
$ \Pi(J) = |J| \ \Psi(J) = O(1)$.
By Proposition~\ref{pro:piS} we have $\Pi(J) > \frac{1}{4 \sqrt{d}}$. This concludes the proof.
\qed
\end{proof}

Denote by $\mathcal{U}_{m,n}(T)$ the family of semi-intervals partitioned by $\mathcal{D}_{m,n}(T)$.

Clearly $\mathcal{V}_{m,n}(T)$ contains $\mathcal{U}_{m,n}(T)$. Indeed every semi-interval $J \in \mathcal{V}_{m,n}(T)$ is a finite union of contiguous semi-intervals belonging to $\mathcal{U}_{m,n}(T)$.

Note that $\mathcal{U}_{m,0}(T)$ is the family of semi-intervals exchanged by $T^m$, while $\mathcal{U}_{0,n}(T)$ is the family of semi-intervals exchanged by $T^{-n}$.

Put $\mathcal{U}(T) = \bigcup_{m,n \geq 0} \mathcal{U}_{m,n}(T)$.
Using Theorem~\ref{theo:pi} we easily deduce the following corollary, which is a generalization of Theorem 2.11 in \cite{BoshernitzanCarroll1997}.

\begin{corollary}
\label{cor:pi}
$\Pi(J) = \Theta(1)$ for every semi-interval $J \in \mathcal{U}(T)$.
\end{corollary}

We are now able to prove Theorem~\ref{theo:quadratic}.

\begin{proofof}{of Theorem~\ref{theo:quadratic}}
By Theorem~\ref{theo:biRauzy2}, every admissible semi-interval can be obtained by a finite sequence $\chi$ of right and left Rauzy inductions.
Thus we can enumerate the family of all admissible semi-intervals.
The conclusion follows easily from Theorem~\ref{theo:bosh} and Theorem~\ref{theo:pi}.
\qed
\end{proofof}

An immediate corollary of Theorem~\ref{theo:quadratic} is the following.

\begin{corollary}
\label{cor:graph}
Let $T$ be a regular interval exchange transformation defined over a quadratic field. Then the extension graph $G(T)$ is finite.
\end{corollary}

\begin{example}
Let $T = T_{\pi, (\beta, 1-\beta)}$
and $S = T_{\pi, (\gamma, 1-\gamma)}$
be two regular $2$-interval exchange transformations, where $\pi = (12)$ is the permutation defined in Example~\ref{ex:rotation}, $\beta = (2 - \sqrt{2})$
and $\gamma = \frac{3 -\sqrt{3}}{2}$.
The equivalence graphs of $T$ and $S$ are represented in Figure~\ref{fig:sqrt}.
Note that the ratio of the lengths of the semi-intervals exchanged by $T$ is $\frac{\beta}{1 - \beta} = \sqrt{2} = 1 + \frac{1}{2 + \frac{1}{2+ \cdots}} = 1 + \frac{1}{1 + \sqrt{2}}$,
while the the ratio of the lengths of the semi-intervals exchanged by $S$ is $\frac{\gamma}{1 - \gamma} = \sqrt{3} = 1 + \frac{1}{1 + \frac{1}{2+ \frac{1}{1+ \frac{1}{2 + \cdots}}}} = 1 + \frac{1}{1 + \frac{1}{1 + \sqrt{3}}}$.

\begin{figure}[hbt]
\centering
\gasset{Nframe=y}
\begin{picture}(64,14)(0,-8)

\node(T1)(0,0){$[T]$}
\node(T2)(20,0){}

\drawedge[curvedepth=3](T1,T2){}
\drawedge[curvedepth=3](T2,T1){}

\node(S1)(50,0){$[S]$}
\node(S2)(64,7){}
\node(S3)(64,-7){}

\drawedge[curvedepth=4](S1,S2){}
\drawedge[curvedepth=4](S2,S3){}
\drawedge[curvedepth=4](S3,S1){}

\end{picture}
\caption{Equivalence graphs of the transformations $T = T_{\pi, (\beta, 1-\beta)}$ (on the left) and $S~=~T_{\pi, (\gamma, 1-\gamma)}$ (on the right).}
\label{fig:sqrt}
\end{figure}
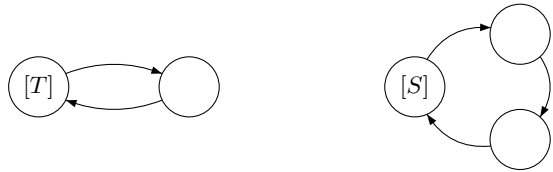
\end{example}

\end{document}